%% file: main.tex
\documentclass[journal]{IEEEtran}
\usepackage[cmex10]{amsmath}
\usepackage{amssymb}
\usepackage{graphicx}
\usepackage{booktabs}
\usepackage{cite}
\usepackage{url}
\usepackage{xcolor}
\usepackage{tikz}
\usetikzlibrary{arrows.meta,positioning,calc}

\newcommand{\krtcommit}{2e319092}
\newtheorem{lemma}{Lemma}
\newtheorem{proposition}{Proposition}
\definecolor{Fcu}{RGB}{200,60,60}
\definecolor{Bcu}{RGB}{50,100,200}
\begin{document}

\title{Whole-Route Planning of a Memory Bus Between Two BGAs: An Exact Model of Crossings and Layer Changes, Negotiated to DRC-Clean Copper}

\author{Andrew~Haas%
\thanks{A. Haas is a Professor in the Department of Physics, New York University, New York, NY 10003 USA (e-mail: andy.haas@nyu.edu).}%
\thanks{The software, the benches and the recipes that reproduce every number in this paper are open source~\cite{krt}; the results are those of its commit \texttt{\krtcommit}.}}

\markboth{IEEE Transactions on Computer-Aided Design of Integrated Circuits and Systems}%
{Haas: Whole-Route Planning of a Memory Bus Between Two BGAs}

\maketitle

\begin{abstract}
Routing the bus between a processor and its DRAM remains largely manual in printed circuit board (PCB) design, and published algorithms optimize proxies such as crossing or layer counts. We present a router that plans every lane of a two-layer BGA-to-BGA bus before any copper is laid, and carries the plan to design-rule-clean copper. Its core is an exact constraint-programming model that chooses, for all lanes at once, where each necessary crossing occurs, in what order---by a rule over every triple of lanes---and where each lane changes layer, with via room; for the ends it is given, it proves its plan optimal within the model. The model reads naturally as physics: worldlines, frustrated Ising spins and contact forces. A linear program places the lanes in millimeters; the constraints it must violate return to the model as cuts, and audit findings as negotiated congestion prices, until the plan passes an audit against the production router's rules or lays clean. The escapes are chosen by an estimate of the same objective, and differential pairs carry their polarity, with an automatic via crossover where the ends disagree. On an Allwinner H3 and a Zynq-7020, each to a DDR3 device, the router completes every rung of two ladders, of up to 51 nets on the H3 (its whole bus) and 43 on the Zynq, connected and DRC-clean, never with more vias than the human designer (80 against 88 and 62 against 100 at the top rungs) and always with less copper, where the toolkit's production router leaves nets open on four of the ten rungs and Freerouting completes none. On the Zynq board itself it completes 45 of the bus's nets with 70 vias to the designer's 105. On generated buses of three, four and six routing layers every run completes, and every one with a known optimum attains it.
\end{abstract}


\section{Introduction}
\IEEEPARstart{A}{} memory bus is the densest wiring on most processor boards: tens of nets leave a ball grid array (BGA) in a fixed order, arrive at another in a different order, and must cross one another on the way. On a two-layer region every crossing forces one of the two nets onto the other layer there, and every layer change is a via---a drill, an annular ring that blocks both layers, and a discontinuity in the signal path. Commercial tools approach it interactively: the designer sequences the bus, plans its breakouts and resolves its crossings, and the tool fills in the route~\cite{cadence2015,pfeil2014}.

The literature treats the problem in stages---escape routing, simultaneous escape routing, bus planning and layer assignment (Section~\ref{sec:related})---each optimizing a proxy such as crossings, layers or a planned via count, evaluated on private industrial cases. Whether a plan can be laid in copper at the board's clearance, and how its via count compares with a designer's, are left open.

This paper starts from an observation that designers, commercial tools and simultaneous escape routing share~\cite{cadence2015,ozdal2006ser}: \emph{the ends decide the braid}. Once each net's escape at the source and destination is fixed, the orders of those escapes around the two arrays fix which pairs of nets must cross. What remains free is where each crossing happens, in what order the crossings occur, and where each lane changes layer. These choices interact, so we decide all of them at once, for every lane, before anything is routed. The contributions are:
\begin{enumerate}
\item An exact constraint-programming (CP) model of the whole route (Section~\ref{sec:solve}) that chooses crossing positions, their order, and every lane's layer changes jointly. Its consistency rule over triples of lanes admits exactly the orders that a wiring diagram of the lanes' permutation can realize. For the ends it is given, it proves its plan optimal within the model.
\item A realization of the plan in board geometry (Section~\ref{sec:copper}): an elastic linear program whose violated constraints return to the CP model as cuts, a snap onto the production router's grid, and an audit against the router's own rules, closed in a negotiation loop that also feeds findings at the ends back to the choice of escapes (Section~\ref{sec:ends}).
\item Differential pairs planned as single lanes with a polarity (handedness) at each end, and an automatic via crossover that swaps P and N where the two ends disagree.
\item An evaluation on two public boards, on benches cut from them and on one board itself, against the human layouts, the toolkit's production router and Freerouting, all graded alike on finished copper, and on generated buses of two to six layers against their exact optimum (Section~\ref{sec:results}).
\end{enumerate}

The models read naturally in physics. Along the route each lane is the worldline of a particle in one space and one time dimension, each crossing a scattering event, and the triple rule the Yang--Baxter relation of three strands. Each worldline carries an Ising spin, its layer, opposite its partner's at every crossing: a via is a domain wall, and frustration forces vias. Placing the lanes is then a problem in statics, whose contact forces are the linear program's dual variables.

\section{Related Work}
\label{sec:related}
\emph{Topological via minimization.} Given the nets' topology, choosing layers to minimize vias is classical. Hsu~\cite{hsu1983} decoupled topology from geometry to minimize vias; Marek-Sadowska~\cite{marek1984} showed the unconstrained problem NP-complete, and Sarrafzadeh and Lee~\cite{sarrafzadeh1989} gave a polynomial algorithm for nets between two shores, where the problem reduces to covering a permutation graph. Constrained via minimization assigns layers to an already routed layout~\cite{pinter1984}: it is NP-complete in general~\cite{naclerio1989}, and where the layout's junctions have degree at most three it is a maximum cut on a planar graph, solvable in polynomial time~\cite{hadlock1975,pinter1984}, which Barahona \emph{et al.}\ solved with the same code as Ising spin-glass ground states~\cite{barahona1988}. Exact algorithms followed for $k$ layers in channels~\cite{rim1989,cong1991} and for a package's via-free redistribution layers~\cite{lin2017}. None of these places the crossings, orders them, or gives vias their room, jointly with the layers.

\emph{PCB escape and bus routing.} Ozdal and Wong chose escape patterns at both ends of a channel jointly, with a longest-path search for each layer's planar subset and negotiation for the rest~\cite{ozdal2006ser}, and Ozdal, Wong and Honsinger minimized the vias the choice causes directly~\cite{ozdal2008via}. Bus planning and layer assignment for dense boards followed~\cite{kong2009,yan2019}; where a bus's pin orders at its two ends disagree, Yan and Wong untangle the twisted nets with single detours near the ends, on one layer and with no vias~\cite{yan2007untangle,yan2009untangle}. Exact models route the escapes at one array~\cite{luo2008,bayless2016} or a two-layer BGA package over a restricted route structure~\cite{tomioka2009}, and differential pairs have escape routers of their own~\cite{yan2010diff}. Wu \emph{et al.}\ unified escape and area routing in one model, searching escape orders by simulated annealing~\cite{wu2026}. Lin \emph{et al.}\ carried simultaneous escape routing through gridless area routing to finish commercial boards that a commercial router could not~\cite{lin2021}. Apart from that flow and the escape routers, this work stops at a plan, and all of it is evaluated on private cases.

\emph{Topology, then geometry.} Rubber-band routers~\cite{dai1991} and the theory of homotopic routing~\cite{leiserson1985} separate a topological route from its embedding. Wen \emph{et al.}\ assign a package's layers by a maximum planar subset and finish its layout by linear programming~\cite{wen2022}. Negotiated congestion~\cite{mcmurchie1995} prices places that were overused before.

\emph{Wiring diagrams.} The orders in which the crossings of a permutation's reduced words can occur are the convex orders of its inversion set~\cite{papi1994}; for the longest permutation their commutation classes form the higher Bruhat order $B(n,2)$~\cite{manin1989}. We are not aware of a model that uses this structure to minimize vias jointly with crossing positions and via room.

\section{The Problem and the Frame}
\label{sec:problem}
\subsection{Setting}
A \emph{bench} is a board region holding a source BGA $S$ (the processor) and a destination BGA $D$ (the DRAM), a set of bus nets between them, and everything else---passives, other nets' stubs, the board edge---as fixed obstacles. The region has two signal layers, F and B (Section~\ref{sec:layers} extends the model to more). Each net leaves $S$ at a \emph{tooth}, the end of a fanout stub from its ball, and arrives at $D$ at a \emph{berth}; each end has a layer. The fanout engine offers every ball a menu of escape moves (a stub on one layer, or a dog-bone via to the other), from which the ends are chosen (Section~\ref{sec:ends}). A differential pair is one \emph{lane} with two legs. A lane's layer change is a \emph{dive}: one via for a single lane, two barrels side by side for a pair. The benches' clearance is 0.1\,mm; we plan and route 5\,$\mu$m above it, at 0.105\,mm, with 0.127\,mm lanes, 0.1\,mm fanout stubs and 0.25\,mm vias (0.15\,mm drill) on the router's grid $g=0.025$\,mm. Two neighboring lanes on one layer must lie at least a \emph{bar} of 0.232\,mm apart, center to center (a lane and the clearance); the least pitch at which lanes are planned is $p=0.252$\,mm, a bar and 20\,$\mu$m.

The objective counts planned vias: one for every layer change of a lane, and five more for every via a net carries beyond two, its stub vias included (a pair counted by its worse leg); the stub vias themselves come with the ends and cost nothing more. A net whose vias are all layer changes thus costs 0, 1, 2, 8 or 14 for 0 to 4 vias. Below a via comes a congestion tie-break. These counts are of the plan, not the board: the stub vias come from the choice of ends (Section~\ref{sec:ends}), and a pair's layer change counts once although it lays two barrels. The result must be routed by the production router, connected, and clean under a design-rule check (DRC) at the board's clearance.

\subsection{The frame and the route coordinate}
\label{sec:frame}
The frame is read off the board alone (Fig.~\ref{fig:frame}). The board is first turned by an exact quarter turn, and mirrored for pairs of the other handedness (both undone on writing), so that $S$ lies west of $D$. Each array's pad box is grown by the median distance its stubs end outside it. $S$'s \emph{far} face is the one away from $D$; $D$'s \emph{facing} face is the one toward $S$, and its far face the one away. A straight \emph{trunk} runs from the center of $S$'s box to that of $D$'s, extended past each by half its diagonal and $L=0.35$\,mm; a point near it is $(s,o)$, arclength along it and offset across it.

The orders are unrolled from \emph{cuts} on the far faces. On $D$'s far face the cut stands in the middle of the widest gap between its berths, the face's ends counted as gap ends, so that berths together near one end go round one side. On $S$'s far face it stands in the gap that makes the fewest crossings. Unrolling the berths round $D$ from its cut, north first, gives the \emph{berth order} $\varphi$; unrolling the teeth round $S$ gives the \emph{launch order} $\lambda$. The ends model hands its cuts to the frame with the plan, so that model and solve share one order.

A lane berthing on $D$'s facing face ends in the trunk; the least $s$ of these berths is the \emph{arrival line} $H_0$. The others ride a \emph{ring}, north or south of $D$'s cut. They leave the trunk at the ring's start $H_k$, stepped back from $H_0$, stacked in berth order a lane's width apart outside the facing face's lanes and the octilinear hull of $D$'s pads, stubs and the copper hugging it. The ring's spine runs taut from there to the hull and then round it, $L$ out, the way its lanes' taut paths sweep. A tooth on $S$'s far face reaches its side face on a \emph{lead-in}, straight out of the face and round $S$'s corner on a $45^\circ$ chamfer, nested with the others of its side; from there the trunk carries it. Each lane's \emph{reference} is its taut path, moved out of both boxes on the lane's own side.

Every lane's route is then one scalar coordinate $u$: on the trunk $u=s$, from the lane's \emph{entry}, the $s$ of its start; on ring $k$, $u=H_k+s'-r_k$, with $s'$ arclength along the ring's spine and $r_k$ the point on it past all of its handoffs. Lanes at the same $u$ in one frame stand across the bus from one another, so a crossing's position is one number for both lanes. We discretize $u$ on a grid $G$ of a fifth of $p$, 0.0504\,mm.

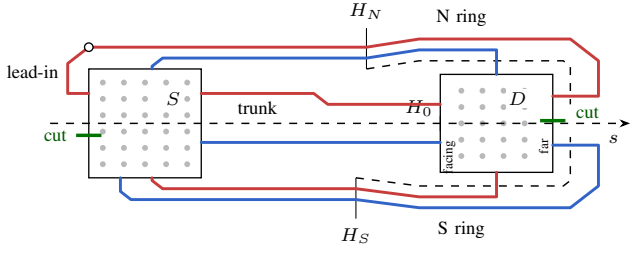
\begin{figure}[t]
\centering
\begin{tikzpicture}[x=0.93cm,y=0.72cm,font=\scriptsize,
  gbox/.style={draw=black,line width=0.5pt},
  spine/.style={dashed,line width=0.55pt},
  lane/.style={line width=1.1pt,line cap=round,line join=round},
  cutm/.style={line width=1.3pt,green!45!black}]
\foreach \x in {0.2,0.5,...,1.45} \foreach \y in {-0.75,-0.45,...,0.8} \fill[gray!55] (\x,\y) circle (1.1pt);
\foreach \x in {5.3,5.6,...,6.35} \foreach \y in {-0.6,-0.3,...,0.65} \fill[gray!55] (\x,\y) circle (1.1pt);
\draw[gbox] (0,-1) rectangle (1.6,1);
\draw[gbox] (5.0,-0.9) rectangle (6.6,0.9);
\node[fill=white,inner sep=1pt] at (1.2,0.45) {$S$};
\node[fill=white,inner sep=1pt] at (6.1,0.45) {$D$};
\draw[spine,-{Stealth[length=4pt]}] (-0.55,0)--(7.7,0) node[below left=1pt]{$s$};
\node[above] at (2.4,0.0) {trunk};
\draw[spine] (3.95,1.0)--(4.7,1.15)--(6.75,1.15)--(6.85,1.05)--(6.85,0.35);
\draw[spine] (3.8,-1.0)--(4.7,-1.15)--(6.75,-1.15)--(6.85,-1.05)--(6.85,-0.25);
\node[above] at (5.3,1.62) {N ring};
\node[below] at (5.3,-1.62) {S ring};
\draw[line width=0.4pt] (3.95,1.0)--(3.95,1.75) node[above]{$H_N$};
\draw[line width=0.4pt] (3.8,-1.0)--(3.8,-1.75) node[below]{$H_S$};
\draw[line width=0.4pt] (5.0,-0.14)--(5.0,0.14);
\node[above left=-1pt] at (5.0,0.0) {$H_0$};
\draw[lane,Fcu] (1.6,0.55)--(3.2,0.55)--(3.4,0.35)--(5.0,0.35);
\draw[lane,Bcu] (1.6,-0.35)--(5.0,-0.35);
\draw[lane,Bcu] (0.9,1.0)--(0.9,1.05)--(1.05,1.2)--(3.95,1.2)--(4.7,1.35)--(5.8,1.35)--(5.8,0.9);
\draw[lane,Fcu] (0.9,-1.0)--(0.9,-1.05)--(1.05,-1.2)--(3.8,-1.2)--(4.7,-1.35)--(5.8,-1.35)--(5.8,-0.9);
\draw[lane,Bcu] (0.45,-1.0)--(0.45,-1.25)--(0.6,-1.4)--(3.8,-1.4)--(4.7,-1.55)--(6.85,-1.55)--(7.25,-1.15)--(7.25,-0.4)--(6.6,-0.4);
\draw[lane,Fcu] (0,0.55)--(-0.3,0.55)--(-0.3,1.1)--(0.0,1.4)--(3.95,1.4)--(4.7,1.55)--(6.85,1.55)--(7.25,1.15)--(7.25,0.5)--(6.6,0.5);
\fill[white,draw=black,line width=0.5pt] (0.0,1.4) circle (1.6pt);
\node[left] at (-0.3,0.95) {lead-in};
\draw[cutm] (-0.18,-0.22)--(0.18,-0.22);
\draw[cutm] (6.42,0.05)--(6.78,0.05);
\node[left,green!35!black] at (-0.18,-0.22) {cut};
\node[right,green!35!black] at (6.8,0.2) {cut};
\node[rotate=90,font=\tiny] at (5.13,-0.62) {facing};
\node[rotate=90,font=\tiny] at (6.47,-0.42) {far};
\end{tikzpicture}
\caption{The frame, schematically. The trunk (dashed, arclength $s$) runs through the grown boxes of $S$ and $D$. Lanes berthing on $D$'s facing face end in the trunk, the first at $H_0$; the others leave it at their ring's start ($H_N$, $H_S$), stacked in berth order, and ride a ring to $D$'s side faces and to its far face, north or south of its cut (green). A tooth north of $S$'s cut reaches the north face on a lead-in, to its start (circle). These six lanes keep one order and cross nowhere. Red is F, blue is B.}
\label{fig:frame}
\end{figure}

A pair of lanes $(a,b)$ with $\lambda(a)<\lambda(b)$ is \emph{inverted} if $\varphi(a)>\varphi(b)$. Every inverted pair must cross; in our model it crosses exactly once, and no other pair crosses. The set $I$ of inverted pairs is the inversion set of the permutation $\pi$ that takes launch order to berth order (Fig.~\ref{fig:braid}). We call the crossings, with the order in which they occur, the \emph{braid}; unlike a braid in topology it records no over and under, which is what the layers decide.

\begin{figure}[t]
\centering
\begin{tikzpicture}[x=0.93cm,y=0.6cm,line width=1.4pt,line cap=round]
\draw[Fcu] (0,0)--(4,0)--(5,1)--(8,1);
\draw[Fcu] (0,1)--(2,1)--(3,2)--(5.5,2)--(6.5,3)--(8,3);
\draw[Fcu] (0,2)--(1.4,2);
\draw[Bcu] (1.4,2)--(2,2)--(3,1)--(4,1)--(5,0)--(5.6,0);
\draw[Fcu] (5.6,0)--(8,0);
\draw[Fcu] (0,3)--(5.0,3);
\draw[Bcu] (5.0,3)--(5.5,3)--(6.5,2)--(7.0,2);
\draw[Fcu] (7.0,2)--(8,2);
\foreach \p in {(1.4,2),(5.6,0),(5.0,3),(7.0,2)} \draw[line width=0.7pt,fill=white] \p circle (3pt);
\foreach \y/\l in {0/1,1/2,2/3,3/4} \node[font=\scriptsize,left] at (-0.05,\y) {$\l$};
\foreach \y/\l in {0/3,1/1,2/4,3/2} \node[font=\scriptsize,right] at (8.05,\y) {$\l$};
\node[font=\scriptsize,above] at (0,3.3) {launch $\lambda$};
\node[font=\scriptsize,above] at (8,3.3) {berth $\varphi$};
\draw[-{Stealth[length=4pt]},line width=0.5pt] (0,-0.8)--(8.3,-0.8) node[right,font=\scriptsize]{$u$};
\foreach \x/\l in {2.5/{t_{23}},4.5/{t_{13}},6/{t_{24}}} {
  \draw[line width=0.5pt] (\x,-0.9)--(\x,-0.7);
  \node[below,font=\scriptsize] at (\x,-0.85) {$\l$};}
\end{tikzpicture}
\caption{Four lanes as a wiring diagram along the route coordinate $u$. The berth order is $\pi=(3,1,4,2)$, so the inverted pairs $\{2,3\},\{1,3\},\{2,4\}$ must each cross once. On the triple $(1,2,3)$ only $\{1,3\}$ and $\{2,3\}$ are inverted, so the triple rule forces $t_{23}<t_{13}$. Red is F, blue is B; circles are vias. Crossing lanes are on different layers. Lane 3 serves its two F partners in one excursion to B, and lane 4 serves one: four vias, equal to the floor $2(K-\mathrm{LIS}(\pi))=2(4-2)$ of Lemma~\ref{lem:floor}.}
\label{fig:braid}
\end{figure}
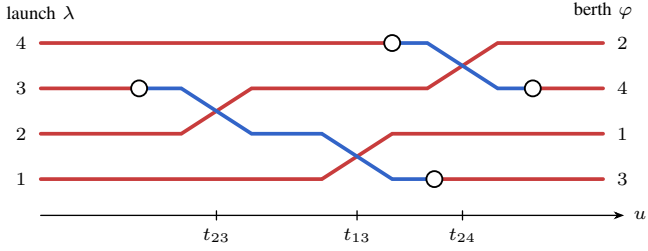

\subsection{Two bounds}
\begin{lemma}[Floor]
\label{lem:floor}
If both ends of every lane are on F, any routing within the frame---lanes that pass round neither array, so that they leave the source in the order $\lambda$ and reach the destination in the order $\varphi$---needs at least $2(K-\mathrm{LIS}(\pi))$ vias, where $\mathrm{LIS}$ is the length of the longest increasing subsequence.
\end{lemma}
\begin{IEEEproof}
Lanes that never change layer stay on F and so cannot cross one another; they form an increasing subsequence of $\pi$. Every other lane leaves F and returns, at least two vias.
\end{IEEEproof}
The floor is the classical maximum planar subset of the two-shore case~\cite{sarrafzadeh1989}, and it is weak: lanes that leave F also cross one another, and an excursion to B cannot serve a crossing with a partner that is itself on B. For a uniformly random permutation the LIS grows as $2\sqrt{K}$, with Tracy--Widom fluctuations of order $K^{1/6}$~\cite{baik1999}, the statistics of the KPZ growth class~\cite{prahofer2000}; real buses are structured, but this is a yardstick for how the floor grows. The next lemma counts a lane's vias exactly, given its partners.

\begin{lemma}[Blocks]
\label{lem:blocks}
Fix the layers of lane $n$'s crossing partners at their crossings, in the order the crossings occur along $n$, and let both ends of $n$ be on F. The fewest layer changes of $n$ is twice the number of maximal runs of consecutive partners that are on F.
\end{lemma}
\begin{IEEEproof}
At each crossing $n$ must be on the partner's opposite layer: on B across a maximal run of F partners, on F across a run of B partners and at both ends. Each F run therefore needs one change into B and one out, and no change is needed elsewhere.
\end{IEEEproof}
The difficulty is therefore not placing a lane's vias given its partners, but choosing every lane's layers and crossing order together---which is what the whole-route model does.

\section{The Whole-Route Solve}
\label{sec:solve}
One CP-SAT~\cite{cpsat} model decides every crossing and every layer change for all lanes at once (Fig.~\ref{fig:beads}).

\begin{figure}[t]
\centering
\begin{tikzpicture}[x=0.9cm,y=0.74cm,font=\scriptsize,
  bead/.style={draw=black,fill=gray!12,rounded corners=1.5pt,line width=0.5pt},
  tie/.style={gray,dashed,line width=0.5pt}]
\fill[green!55!black,opacity=0.22,rounded corners=1pt] (3.2,-0.19) rectangle (6.8,0.19);
\node[green!35!black,font=\tiny] at (6.15,-0.42) {window of $t_{13}$};
\draw[Fcu,line width=1.6pt] (0,0)--(8,0);
\draw[Fcu,line width=1.6pt] (0,1)--(8,1);
\draw[Fcu,line width=1.6pt] (0,2)--(1.4,2);
\draw[Bcu,line width=1.6pt] (1.4,2)--(5.6,2);
\draw[Fcu,line width=1.6pt] (5.6,2)--(8,2);
\draw[Fcu,line width=1.6pt] (0,3)--(5.0,3);
\draw[Bcu,line width=1.6pt] (5.0,3)--(7.0,3);
\draw[Fcu,line width=1.6pt] (7.0,3)--(8,3);
\draw[tie] (2.5,1)--(2.5,2);
\draw[tie] (4.5,0)--(4.5,2);
\draw[tie] (6.0,1)--(6.0,3);
\draw[bead] (2.2,0.86) rectangle (2.8,1.14);
\draw[bead] (2.35,1.86) rectangle (2.65,2.14);
\draw[bead] (4.2,-0.14) rectangle (4.8,0.14);
\draw[bead] (4.35,1.86) rectangle (4.65,2.14);
\draw[bead] (5.7,0.86) rectangle (6.3,1.14);
\draw[bead] (5.85,2.86) rectangle (6.15,3.14);
\node[font=\tiny,fill=white,inner sep=0.6pt] at (2.5,1.5) {$t_{23}$};
\node[font=\tiny,fill=white,inner sep=0.6pt] at (4.5,0.5) {$t_{13}$};
\node[font=\tiny,fill=white,inner sep=0.6pt] at (6.0,2.5) {$t_{24}$};
\foreach \p/\l/\w in {(1.4,2)/{c_{3,1}}/above, (5.6,2)/{c_{3,2}}/below, (5.0,3)/{c_{4,1}}/above, (7.0,3)/{c_{4,2}}/above} {
  \draw[line width=0.6pt,fill=white] \p circle (3.2pt);
  \node[font=\tiny,\w=2.5pt] at \p {$\l$};}
\foreach \y/\l in {0/1,1/2,2/3,3/4} {
  \node[left] at (-0.1,\y) {$\l$};
  \node[right,font=\tiny,Fcu] at (8.05,\y) {F};}
\node[font=\tiny,align=left,anchor=south west] at (-0.3,3.3) {bead: a crossing (full pitch on the lane that stays, shorter on the mover)\\circle: a layer change};
\draw[-{Stealth[length=4pt]},line width=0.5pt] (0,-0.75)--(8.3,-0.75) node[right]{$u$};
\end{tikzpicture}
\caption{The whole-route model as beads on tracks, for the four lanes of Fig.~\ref{fig:braid} (all ends on F). Each lane's route coordinate $u$ is a track, drawn in launch order. A crossing variable $t_{ab}$ is one bead that sits on both lanes' tracks at once (dashed tie), and intervals on one track may not overlap; each crossing slides only within its window (green, shown for $t_{13}$). Layer changes $c_{n,k}$ are beads too (circles), and between them a track's color is its layer. The model requires tied tracks to differ in color at their bead and each lane's number of changes to have the parity of its end layers, and the triple rule orders $t_{23}$ before $t_{13}$. The solver slides every bead at once and proves the plan optimal.}
\label{fig:beads}
\end{figure}
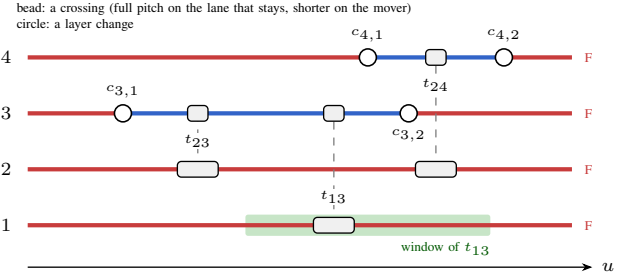

\emph{Crossings.} For each $(a,b)\in I$ an integer $t_{ab}\in[\ell_{ab},h_{ab}]$ on the grid $G$ is the crossing's position along the route. Its window opens past both lanes' entries (a pair's end connector, inside which nothing may cross it) and, on two layers, past a band of 0.65\,mm in front of $S$ where the teeth stand; it closes before the earlier lane leaves the trunk, at its landing or its ring's start, or, where both ride one ring, before the earlier landing on it. The \emph{triple rule}: for every $i<j<k$ in launch order,
\begin{itemize}
\item if $ij,ik,jk\in I$: $t_{ij}<t_{ik}<t_{jk}$ or $t_{jk}<t_{ik}<t_{ij}$ (one Boolean);
\item if $ij,ik\in I$ and $jk\notin I$: $t_{ij}<t_{ik}$;
\item if $ik,jk\in I$ and $ij\notin I$: $t_{jk}<t_{ik}$.
\end{itemize}
(The case $ij,jk\in I$, $ik\notin I$ cannot occur.) The rule is a known characterization:
\begin{proposition}
\label{prop:triple}
A linear order of $I$ is the order in which the crossings of some wiring diagram (reduced word) of $\pi$ occur if and only if it satisfies the triple rule.
\end{proposition}
\begin{IEEEproof}[Sketch]
In a wiring diagram two lanes cross only when adjacent, so after any prefix of crossings the lanes stand in some permutation's order and the prefix is its inversion set; conversely, a pair whose addition keeps the prefix an inversion set is adjacent and can cross next. A set of pairs is an inversion set if and only if, on every triple, it is closed ($ij,jk$ imply $ik$) and co-closed ($ik$ implies $ij$ or $jk$); requiring this of every prefix gives exactly the three cases above, Papi's characterization of convex orders~\cite{papi1994}.
\end{IEEEproof}
Crossings of disjoint pairs may coincide---they commute---while crossings that share a lane are kept apart by a per-lane no-overlap constraint on intervals starting at each crossing: $p$ long on the lane that \emph{stays} and $2G$, the footprint of a steep sweep, on the lane that \emph{moves} across, a Boolean per crossing choosing which. Two crossings of a pair lane that pass other lanes in opposite directions are kept two turning runs and a pitch apart.

\emph{Worldlines.} Read $u$ as time and the lane order as position: lanes are worldlines, and a crossing is a two-body scattering event. Three mutually crossing lanes can scatter in only two orders, $t_{ij}<t_{ik}<t_{jk}$ or its reverse---the two sides of the Yang--Baxter relation~\cite{baxter1972}, or the third Reidemeister move. The triple rule's Boolean is the choice between them.

\emph{Layer changes.} Each lane $n$ has up to three change positions $c_{n,1..3}$ with activity literals $a_{n,k}$; active changes are ordered, a via's room $D_v$ apart (0.65\,mm for a single, a pitch more for a pair), and kept out of the band in front of $S$ and out of a room at each end; their parity is fixed by the ends:
\begin{equation}
\textstyle\bigoplus_k a_{n,k} \;=\; \ell^{S}_n \oplus \ell^{D}_n ,
\end{equation}
where $\ell^S_n,\ell^D_n\in\{0,1\}$ are the end layers. For each crossing $(a,b)$ and each change $k$ of lane $a$, a literal $\beta_{a,k,ab}$ states that the change lies before the crossing, enforced with clearance: $c_{a,k}\le t_{ab}-h$ or $c_{a,k}\ge t_{ab}+h$, where $h$ is half of $D_v$, a pair's pitch more when the crossing lane is a pair, and a tenth more when $a$ stays. The layer of $a$ at the crossing is $\ell^S_a\oplus\bigoplus_k(a_{a,k}\wedge\beta_{a,k,ab})$, and crossing lanes must differ:
\begin{equation}
\textstyle
\ell^S_a\oplus\ell^S_b\oplus\bigoplus_k (a_{a,k}\wedge\beta_{a,k,ab})\oplus\bigoplus_k (a_{b,k}\wedge\beta_{b,k,ab}) \;=\; 1 .
\end{equation}
Changes of lanes that can be neighbors---adjacent in $\lambda$ or $\varphi$, or crossing each other---are staggered along $u$, within each frame, by $\max(d,\sqrt{v_v^2-p^2})$, where $d$ is the via's diameter and $v_v=0.38$\,mm the least distance between two vias' centers and a grid step, so that two vias one pitch $p$ apart across the bus still meet it.

\emph{Spins.} Give each worldline an Ising spin, its layer. Equation~(2) is then a hard antiferromagnetic bond at every crossing, a via is a domain wall along a worldline, and the objective counts domain walls. A cycle of must-differ constraints with an odd number of links cannot be satisfied without a wall---it is frustrated~\cite{toulouse1977}---and frustration is what forces vias. With the crossing order fixed this is the planar spin-glass ground state that Barahona \emph{et al.}\ solved as a maximum cut alongside via minimization~\cite{barahona1988}, and Lemma~\ref{lem:blocks} counts the domain walls of one worldline in the field of its partners. The whole-route solve also chooses the order of the scattering events.

\emph{Built-in cuts} keep a lane's changes off stretches within a via's reach of a pad of another part, or of another net's tooth, that its reference passes within a pitch of, and a pair's off its ring's corners. A cut of half-width $w$ at $u$ is a clause on each change: $c_{n,k}\le u-w$, $c_{n,k}\ge u+w$, or the change is not made ($a_{n,k}=0$).

\emph{Objective.} With $o_n\ge 0$ and $o_n\ge s_n+\sum_k a_{n,k}-2$ counting the vias of lane $n$ beyond two ($s_n$ its stub vias, the larger leg's for a pair, less a tie via that joins a ball to a pad of its own net under it),
\begin{equation}
\min\; W_O\sum_n o_n \;+\; W_V\sum_{n,k} a_{n,k} \;+\; W_S\sum_j b_j \;+\; H ,
\end{equation}
with $W_V=10^6$, $W_O=5W_V$, $W_S=3W_V$ for each soft cut $j$ broken (Section~\ref{sec:loop}), and $H$ the history term of Section~\ref{sec:loop}, a small fraction of a via for each place it prices. Preferring at most two vias per net is a price, not a cap: a cap is infeasible on many instances, and a lexicographic preference would buy one net's third via back with any number of other vias.

\emph{Proving it.} CP-SAT runs on four workers in deterministic interleaved batches, clause sharing off, so every run is the same. A first \emph{probe} restricts its full-problem searches to those that raise the bound---core-based search and two linear relaxations---for at most ten batches; CP-SAT's default portfolio dilutes these among five plan-finding searches, and left the 51-net rung's first solve unproved where the probe proves it. If the probe's plan is not proved in whole vias, plan-finding searches continue from it, with core-based search and the probe's bound as a constraint, for at most a hundred batches. A re-solve in the loop starts from the previous plan and only adds cuts and history, so a proved first solve's via terms are a valid floor for it, and a plan at that floor is proved at once. A plan the search cannot prove is still kept and laid, its nets over two vias returned to the ends; where there is no plan at all, a relaxed model names the crowded lanes, and they return to the ends. On both ladders every first solve was proved (Section~\ref{sec:cost}).

\emph{Scope.} The proof is relative to the model: the ends are fixed, each inverted pair crosses once (non-reduced braids are excluded), a lane changes layer at most three times (no plan on either ladder uses more than two), positions lie on $G$, and the built-in and loop cuts are heuristic. What is proved is the least objective in whole vias for the given ends: the layer changes, five more per via a net carries over two, and three per soft cut broken. The board's via count also depends on the ends, whose stub vias are about half of it, and a better choice of ends can beat a proved plan.

\section{Choosing the Ends}
\label{sec:ends}
Because the ends decide the braid, the escapes at both arrays are chosen with an estimate of the solve's objective, to which the search adds prices of its own: the lanes' length at a fixed rate per via, trunk congestion and crossings, and the audits' feedback (Section~\ref{sec:loop}). For a candidate assignment of teeth and berths we estimate the route's vias from the ends alone:
among lanes whose ends share a layer, inverted pairs on one layer must be separated by lanes that dive, and since their conflict graph is a permutation graph the cheapest cover is the complement of a maximum-weight increasing chain, exact per layer; a lane whose ends differ changes once, and the partners it would meet on the wrong layer form a bipartite graph covered exactly by a minimum cut; conflicts between such lanes are resolved greedily. Congestion adds each lane's load on the trunk (the room its crossings and changes take), priced by the square of its excess over half, and a fifth of a via per trunk crossing. The crossing price exists because the solve's effort follows its crossings, not its vias. The search moves one tooth or berth at a time (ejecting one conflicting lane where needed), then iterates from seeded random kicks. The best few states are then re-ranked \emph{exactly} on their orders by a lighter form of the solve's model, with no crossing positions: a Boolean per triple of mutually crossing lanes, one transitive order of each lane's crossings, and each layer change before or after each of them. Since of $c$ mutually crossing lanes on two layers at least $c-2$ change layer, that bound is added for every maximal clique of the crossing graph, and CP-SAT searches core-guided, since nearly every lane has no change. Each exact ranking stops at a work limit and the search is local, so the chosen ends are not proved best.

\emph{Street berths.} A ball beside the unpopulated band of a DDR3 package may dive in a lane of the band and leave along it on the other layer, as the human designers do. Offering such berths trades a few vias for fewer crossings, and so for a quicker solve.

\section{From Plan to Copper}
\label{sec:copper}
A plan becomes copper in stages, each with a margin over the next: the geometry and the polish keep a grid step more room than the snap and the audit. The geometry and the polish place the lanes in millimeters, and an audit gates this smooth plan; the snap then lays the pairs on the router's grid, the polish fits the singles round them, and the snap lays the singles, each step audited; a last audit, with a lint of the router's shape rules, gates the snapped plan, and the production router lays it.

\subsection{Geometry}
\label{sec:geometry}
A single linear program places every lane: an offset per lane per column, columns every 0.1\,mm along the trunk and both rings, a ring lane's two pieces joined where it leaves the trunk. The solve fixes each column's order and each lane's layers; the program moves lanes and vias across the spine only. Constraints keep same-layer neighbors a pitch apart, a via's room round every change (chords of a disc at neighboring columns), every lane inside the board and off the pad boxes, and the slope's change within one per column, a $45^\circ$ turn from a run along the spine; the ends are pinned, with a few columns beside them held parallel to the spine. Costs are the added length---under-estimated by tangent cuts, so that a slope below $13^\circ$ costs nothing and lanes take a \emph{comfortable} pitch where there is room---the bends, and a hinge penalty on neighbors closer than the comfortable pitch $2p$. Every hard constraint except the pins, the holds and the joins is elastic, each row's slack priced at $W_H=10^4$ per millimeter, so the program always has a solution and reports exactly where it had to break a rule.

A second pass holds each lane to one side of every \emph{island}---a part's pads, joined with any pads closer than a lane can pass between. The choice of side is non-convex, so it is made combinatorially before the pass, by neither the LP nor the CP-SAT model: at each island the lanes, in the order the CP-SAT solve fixed there, are split at the index that the first pass's offsets make cheapest, a split leaving either side too little room for its lanes priced out, and the island is inserted into the order at that split. The sides are thus chosen consistent with the crossing plan, and the program stays convex; the polish may still flip a lane it cannot hold on its side (Section~\ref{sec:polish}). The second pass also corrects each pitch for slope by tangent cuts at the first pass's slopes. Most of its rows are such cuts and few bind, so they are added lazily: the pass is solved with those within 0.1\,mm of binding at the first pass's answer, then again with those near binding at any answer that breaks one, each solve warm-started from the last one's basis, until none is broken. The result is the whole program's optimum. The constraints paid in the second pass become cuts for the solve: a lane that could not be kept off an island may not cross anyone over the island's span (an \emph{island cut}); a change that could not be given its room may not lie there (a \emph{via cut}).

The program is large---on the 51-net rung 7\,228 lane offsets and 150\,753 rows in pass one, and 246\,648 rows in pass two, of which its four lazy solves use 148\,257---so we pass its dual to HiGHS~\cite{huangfu2018} through SciPy, reading the offsets back as the multipliers of the dual's balance equations; the two passes take 26 and 124\,s.

\emph{Forces.} The program's optimality conditions are those of a mechanism at rest (Fig.~\ref{fig:lp}). Each lane point is a bead that slides on its column's rail, the length term is a string---slack near the spine, its tension rising toward unity as it leans---and the rules are contacts: struts between neighbors, walls at the pads and the board edge. A rule's dual value is the force in its contact, and complementary slackness is the contact law---a strut pushes only while it touches. An elastic row is a rigid-plastic strut that carries up to $W_H$ and then yields. The static and kinematic theorems of plastic collapse are this duality~\cite{charnes1959}, and the linear program that tests a packing of hard spheres for jamming solves for the same contact forces~\cite{donev2004}. Only a small fraction of the constraints carry a force, and they name the rules that hold each lane where it is; a constraint whose slack is paid is a strut that yields, and those of the static, via and dive kinds return to the solve as cuts.

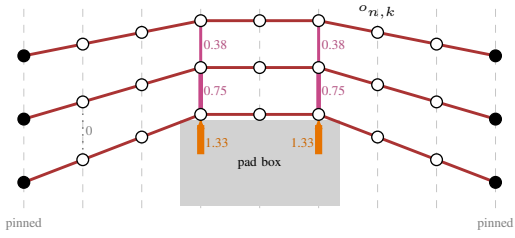
\begin{figure}[t]
\centering
\input{figs/lp_cartoon.tex}
\caption{A small instance of the geometry LP, solved: three lanes on one layer, pinned at both ends (black), pass over a pad box. Each lane point slides on its column's rail (dashed); the lanes are strings at unit tension (the toy also pulls each weakly toward its straight reference). The binding rules carry the dual values drawn: the box (orange) pushes the lowest lane up where the lanes bend, and pitch struts (magenta) pass the load up the stack. Every other rule, like the strut near the pins (dotted), carries no force. Forces are in units of the string tension.}
\label{fig:lp}
\end{figure}

\subsection{Polish and snap}
\label{sec:polish}
The \emph{polish} meets the audit's own measures in board coordinates---pitch, via room, clearance to pads as drawn, stub joins and turn shapes---by small vertex moves, one linear program per round within a trust box of $\pm0.05$\,mm, every rule elastic. A rule enters a round's program while it stands within 0.1\,mm of its bar; a lane or via found inside a pad is pushed out of the pad's island, not toward its nearest edge; folds and notches the audit would flag are straightened between cycles of rounds; and the state kept is the best any round reached, not the last. A pair's end runs, and the straight runs the pair router needs either side of each dive, are laid straight first and held. A lane that cannot be held on its side of an island is flipped, and the geometry runs again; a via that cannot be given its room becomes a via cut.

\label{sec:snap}The \emph{snap} makes the smooth plan octilinear on the router's grid, one lane at a time, by an A$^*$ search in a strip $2p$ either side of each lane's smooth line that prices length, bends and distance from the line, each via within $2p$ of the plan's. Cells are blocked as the router will block them: by the router's own map of static copper, by the real copper of lanes already laid, and by the \emph{shares} of lanes not yet laid---each keeps the side of every midline nearer its own line, less half a bar---so lanes laid first cannot take the room later ones need. Pairs go first, moving exactly as the pair router moves ($45^\circ$ turns, a straight run after each turn, a via only on a straight run); they are then held while the polish fits the singles round them, and the singles are snapped. Up to two sweeps then lay every lane again against the others' real copper. A pair that cannot be laid returns a via cut at its nearest change.

\subsection{Differential pairs}
\label{sec:pairs}
A pair is one lane of two slots' width. Its \emph{hand} at an end is the side of its travel on which P lies. A pair whose hands at its tooth and berth disagree must swap its legs somewhere; the solve then requires it to change layer at least once, and the snap lays a \emph{crossover} at its first dive, as a designer would: the first leg steps out at $45^\circ$ to its barrel, dives, and jogs across the pair's line on the new layer; the second jogs at $45^\circ$ over it to a barrel just beyond, dives, and steps onto the first leg's old line. Both barrels lie on one side of the pair, staggered by the fewest whole grid steps that keep a via's pitch and each jog's clearance, and every stage knows the crossover's footprint before it is laid. At each end the legs run, through an \emph{end connector}---two legs from the pair's tips, turning at most $45^\circ$ at their knees and onto the pose's heading---onto a \emph{pose}, a grid point and heading from which the pair router takes over, so plan and router share one end.

\subsection{Audit and route}
\label{sec:route}
The audit checks the plan against the router's rules as the router will see them (pads as drawn, each terminal joined as the router joins it), waiving only shortfalls within a quarter grid step, which the router's own rounding moves; the snapped plan must also pass a lint of the router's shapes (vertices on the grid, $45^\circ$ angles, a via at every layer change, the pair router's turns and straights). The gate passes a plan only when it is complete and every count is zero. The production router~\cite{krt} then lays the pairs, then the singles, each lane confined to its \emph{band}---the router cells within half a grid step of its snapped line, with its end and via cells. A pair the router refuses in its band is tried in wider ones, and free of the plan last, after every lane has had its band, so that it cannot take the room the plan gave the singles. The audit models the router's rules but is not the router: a plan that passes the gate can still have a lane refused, which is why the board, not the plan, is graded. Independent checks grade connectivity and DRC on the written board.

\subsection{The negotiation loop}
\label{sec:loop}
Every stage is fed a measurement of the one before (Fig.~\ref{fig:loop}); a round ends at the first stage that fails, with what that stage found. Every audit's findings return to the solve as \emph{history}: every place a round found the plan short is priced, as a negotiated router prices a place overused before~\cite{mcmurchie1995}, by how many audits found it so: about $10^{-4}$ of a via for a crossing there and $10^{-3}$ for a layer change, for each audit. History thus reorders plans of equal vias; the solve's proof assumes its total stays below a via. The geometry's and the polish's paid constraints and the pair snap's failures return as cuts, and a finding at a part on one layer only as a soft \emph{layer cut} that holds the lane on the other layer there; side flips return to the geometry. When a re-solve proves no plan under its hard cuts, every hard cut so far becomes a soft cost of three vias. The loop stops when the gate passes, after six rounds or two that do not improve on the best, or when a finding stands at the ends, which only the fanout can move. After a loop that does not pass, its findings at the ends are priced in the ends model, doubling when a later round names the same end again, and the next fanout round chooses again, freeing only the ends named so far.

\emph{Never nothing.} A loop that does not pass is still laid, from its last snapped plan or else its best round's, and a lane that cannot be laid stays open; a round that lays connected and DRC-clean ends the run, whether or not its loop passed and its solve proved the plan. A round that leaves nets open sends them back to the fanout, and the run's result is its best round: no DRC violation first, then the fewest open nets, then vias. This is a counterexample-guided loop in the spirit of combinatorial Benders decomposition~\cite{hooker2003,codato2006}; unlike it, a cut here generalizes one geometric failure and is not proved valid, so the loop's result is optimal only for the model it ends with.

\emph{Sectors.} The geometry's non-convexity lives entirely in discrete choices made before it---the crossing order, the layers and the side of each island, together a homotopy class of the whole bus. Within one such topological sector, and with the first pass's linearization fixed, the geometry is convex. The loop is a search over sectors: the solve proposes one, the geometry relaxes within it, and the constraints that yield tell the solve which sectors to leave.

\begin{figure}[t]
\centering
\begin{tikzpicture}[font=\scriptsize,node distance=1.05cm and 0.55cm,
  box/.style={draw,rounded corners=2pt,align=center,minimum height=0.72cm,minimum width=1.3cm,inner sep=2pt,fill=gray!8},
  fb/.style={-{Stealth[length=3.5pt]},dashed,line width=0.5pt},
  fw/.style={-{Stealth[length=3.5pt]},line width=0.6pt},
  lb/.style={font=\tiny,inner sep=1pt}]
\node[box] (e) {ends\\fanout};
\node[box,right=of e] (s) {solve\\CP-SAT};
\node[box,right=of s] (g) {geometry\\LP};
\node[box,right=of g] (p) {polish\\LP};
\node[box,below=of p] (a) {audit\\smooth};
\node[box,below=of g] (n) {snap A$^*$\\pairs, singles};
\node[box,below=of s] (l) {audit\\and lint};
\node[box,below=of e] (r) {route\\router};
\node[box,below=0.45cm of r,fill=white] (c) {board:\\DRC, conn.};
\draw[fw] (e)--(s); \draw[fw] (s)--(g); \draw[fw] (g)--(p); \draw[fw] (p)--(a); \draw[fw] (a)--(n); \draw[fw] (n)--(l); \draw[fw] (l)--(r); \draw[fw] (r)--(c);
\draw[fb] (p.north) to[bend right=45] node[lb,below]{flips} (g.north);
\draw[fb] (g.north) to[bend right=45] node[lb,below]{cuts} (s.north);
\draw[fb] (p.north east) to[out=120,in=60,looseness=0.75] node[lb,above]{via cuts} (s.north west);
\draw[fb] (a.north west) -- node[lb,sloped,above,pos=0.55]{history, layer cuts} (s.south east);
\draw[fb] (n.north west) -- node[lb,sloped,below,pos=0.45]{cuts} ($(s.south)+(0.25,0)$);
\draw[fb] (l.north) -- node[lb,left]{history} (s.south);
\draw[fb] (r.north) -- node[lb,left]{open} (e.south);
\draw[fb] (s.south west) to[bend left=50] node[lb,below]{crowded} (e.south east);
\draw[fb] (a.south) -- ++(0,-1.6) -| node[lb,pos=0.2,below]{findings at the ends} ($(e.west)+(-0.4,0)$) -- (e.west);
\end{tikzpicture}
\caption{The whole route. Solid arrows pass the plan forward: an audit gates the smooth plan before the snap, and a last audit with a lint gates the snapped plan before the router lays it. Dashed arrows are the negotiation: lanes with no room beside an island return to the geometry as side flips; rules the geometry or the polish had to break, and lanes the snap could not lay, return to the solve as cuts; the audits' findings return as history, and those at a one-layer part as layer cuts; findings at the ends, open nets and crowded lanes return to the choice of ends.}
\label{fig:loop}
\end{figure}
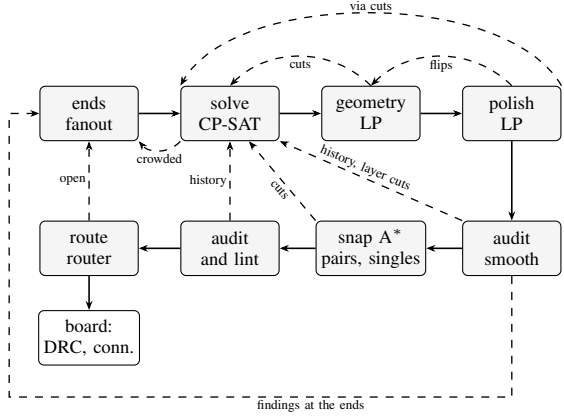

\begin{figure*}[t]
\centering
\includegraphics[width=\textwidth]{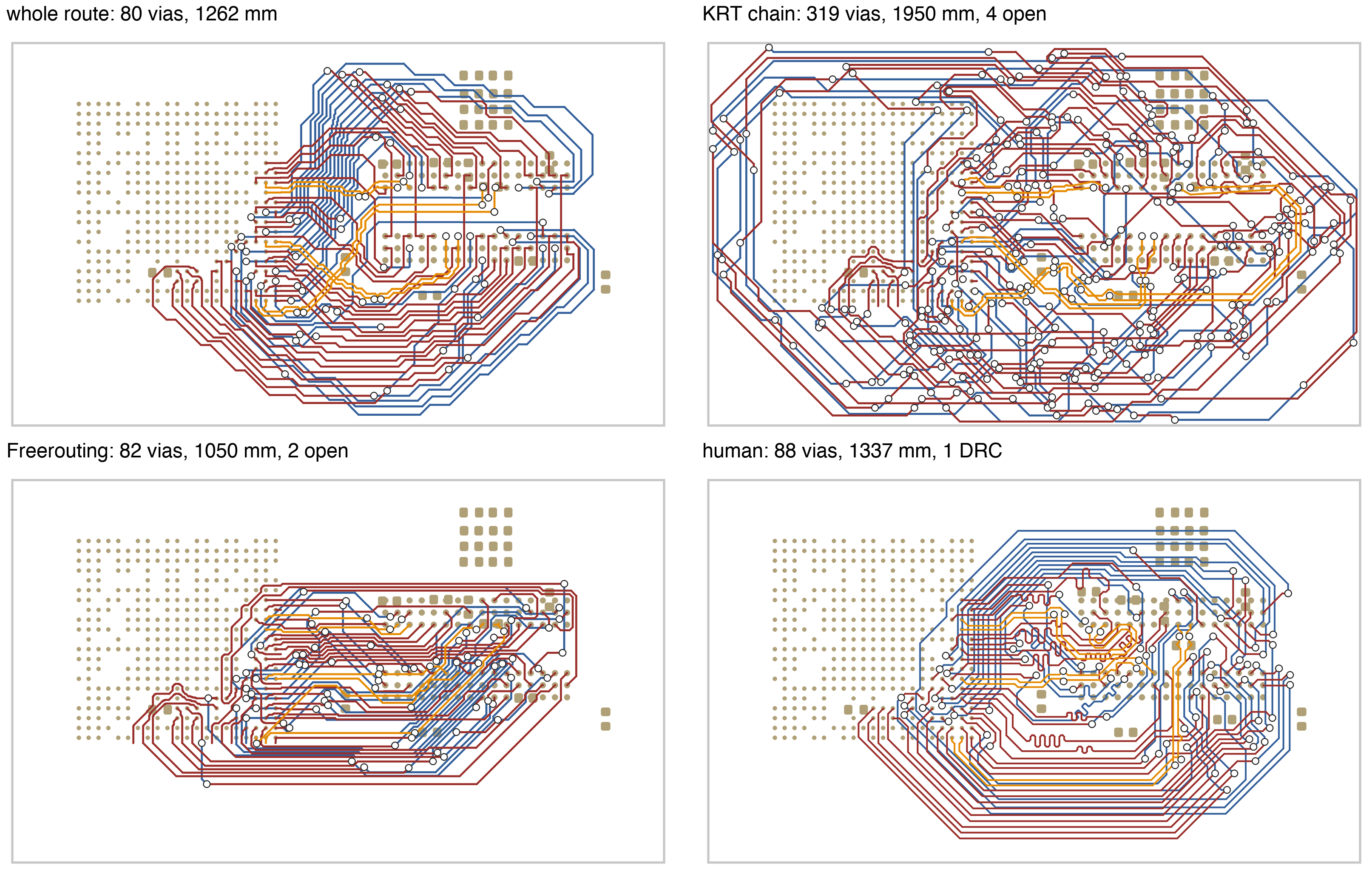}
\caption{The 51-net H3 rung by the whole route on its own ends, the toolkit's production router without the bus step, Freerouting, and the human, only the run's nets drawn; each panel is labeled with its vias, copper, open nets and DRC violations. F is red, B blue, the three differential pairs orange; every via is drawn as the same marker. The source BGA is at left, the DDR3 at right.}
\label{fig:k51}
\end{figure*}

\begin{table*}[t]
\centering
\caption{Vias, copper (mm), open nets and DRC violations of the run's nets on the finished board: the whole route, the toolkit's production router without the bus step (KRT chain), Freerouting 2.4.1, and the human's board, every board graded the same way. The Zynq's last row, Full, is its bus on its own board (Section~\ref{sec:board}), on the 45 nets graded there.}
\label{tab:results}
\setlength{\tabcolsep}{3.6pt}
\newcommand{\rsep}{\hspace{9pt}\vrule\hspace{9pt}}
\setlength{\aboverulesep}{0pt}\setlength{\belowrulesep}{0pt}\renewcommand{\arraystretch}{1.15}
\begin{tabular}{l rrrr @{\rsep} rrrr @{\rsep} rrrr @{\rsep} rrrr}
\toprule
& \multicolumn{4}{c@{\rsep}}{whole route} & \multicolumn{4}{c@{\rsep}}{KRT chain} & \multicolumn{4}{c@{\rsep}}{Freerouting} & \multicolumn{4}{c}{human} \\
\cmidrule(l{4pt}r{20pt}){2-5}\cmidrule(l{4pt}r{20pt}){6-9}\cmidrule(l{4pt}r{20pt}){10-13}\cmidrule(l{4pt}r{2pt}){14-17}
rung & vias & mm & open & DRC & vias & mm & open & DRC & vias & mm & open & DRC & vias & mm & open & DRC \\
\midrule
\csname @@input\endcsname figs/table1.tex 
\bottomrule
\end{tabular}
\end{table*}

\section{Results}
\label{sec:results}
\subsection{Benches and grading}
\emph{H3.} An Allwinner H3 (0.65\,mm ball pitch) to a DDR3 device on a public reference design~\cite{h3board}, whose designer routed the bus on its two layers. Its bench carries 51 bus nets, among them three differential pairs (two data strobes and the clock); the ladder's rungs K15--K51 add whole groups of nets, tightest first.
\emph{Zynq.} The Zynq-7020 (CLG400) and its DDR3 on a public design~\cite{zynqboard}, whose designer routed the bus on the outer layers only (the inner layers are planes), every via F to B. We strip its routing, fan out the processor and keep the two outer layers, by a recipe from the board file in the repository~\cite{krt}; the ladder's rungs K18--K43 carry both strobe pairs.

Each result is written as a board and graded independently: connectivity of every bus net, DRC at the board's clearance of 0.1\,mm, and a census of every via and millimeter of the run's nets, stubs included. The other routers' boards and the human's are graded the same way. Runs used CP-SAT from OR-Tools 9.15 and HiGHS through SciPy 1.16, on Python 3.14, on an Apple M3 laptop (8 cores, 8\,GB), at the toolkit's commit \texttt{\krtcommit}.

\subsection{Against the human}
Table~\ref{tab:results} gives every rung of both ladders on our own ends. Every rung routes every lane in its band at once, every net connected, DRC-clean, never with more vias than the human and always with less copper; every rung but H3 K28 and K35, whose first fanout the fanout's own audit refused, does so in its first fanout round. At the top of the H3 ladder the route uses 80 vias to the human's 88 (Fig.~\ref{fig:k51}); on the Zynq, 62 to 100. Part of the copper difference is the human's length matching.

\subsection{Against other routers}
Two other routers ran on the same benches, nets and design rules, graded the same way (Table~\ref{tab:results}): the production chain of the toolkit the whole route belongs to~\cite{krt}, without the bus step---its BGA fanout at the destination, its differential-pair router and its negotiated A* router---and Freerouting 2.4.1~\cite{freerouting}, an open-source autorouter, through KiCad's Specctra export, with its default settings but two: at most 100 passes, which no run reached, and its score's charge for a clearance violation, a fifth of an unrouted connection, lowered a thousandfold. Its own check counts the benches' fixed copper, legal at the board's clearance, as up to 106 violations, mostly where its octagons stand out past round pads and track ends; at the default charge these held the smaller Zynq rungs' score at zero and stopped them early. Freerouting routes each leg of a differential pair as an independent net.

The production chain completes six of the ten rungs and leaves one to four nets open on the others, with 2.0 to 4.7 times the whole route's vias (319 against 80 at H3 K51) and more copper on every rung; at Zynq K43 its board also has seven clearance violations. Freerouting completes none: it leaves 1 to 5 nets open, and at Zynq K33 its board also has seven DRC violations. On the nets it does connect it can use fewer vias than the whole route---44 against 54 at H3 K41, with four nets open---because its lanes enter the destination's ball field from the near face, where our ends escape every ball to the array's edge (Section~\ref{sec:limits}). In wall time on the same laptop, the production chain takes 5 to 59\,s a rung and Freerouting 1.6 to 3.9 minutes; the whole route takes 18 to 290\,s, 3.6 to 6.8 times the production chain's time and from a fifth of Freerouting's to 1.5 times it. On the Zynq bus on its own board (Section~\ref{sec:board}) the bus step takes 382\,s, the production chain 78\,s and Freerouting 164\,s.

\subsection{On the board itself}
\label{sec:board}
The Zynq bench cuts the bus out of its board at the paper's uniform sizes. The toolkit's \emph{bus step} instead routes it on the board itself, as the production chain hands the board on: planes poured on the inner layers, the board's own design rules, both arrays fanned out by the step, and the lanes on F and B, where the designer routed them. The step takes 46 of the nets between the arrays, leaving the DDR reset's pull-down stub to the rest of the chain; it refuses the clock pair, which passes through a series termination, and the column-address strobe, for which its fanout found no escape. The production chain and Freerouting ran on the same board, nets, layers and rules, and every board, the human's included, was graded on the 45 nets other than the reset at the board's clearance (Table~\ref{tab:results}, Fig.~\ref{fig:board}). The whole route completes them in its first round, with 70 vias to the human's 105 and less copper; the production chain leaves one net open and two clearance violations, with 162 vias, and Freerouting two nets open with 70.

\begin{figure*}[t]
\centering
\includegraphics[width=\textwidth]{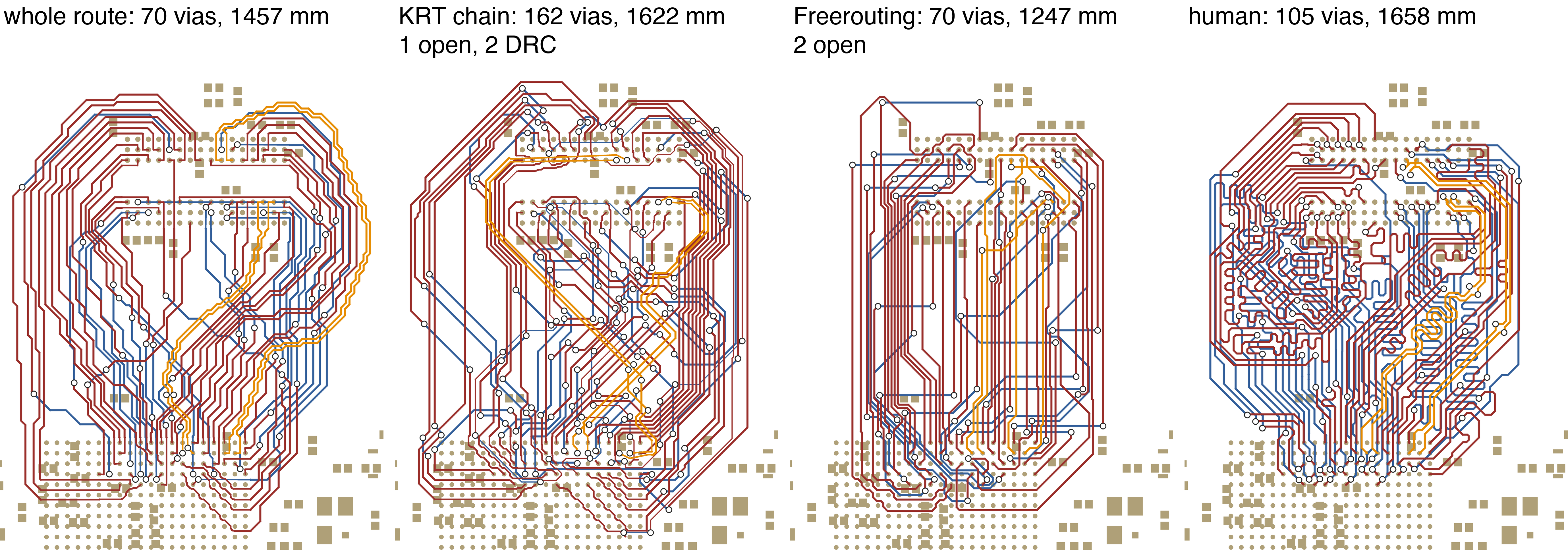}
\caption{The Zynq DDR3 bus on the board itself, drawn as Fig.~\ref{fig:k51}: the whole route's bus step, the toolkit's production router, Freerouting and the human, the 45 nets graded on each, each panel labeled with its vias, copper, open nets and DRC violations. The processor is at the bottom, the DDR3 at the top.}
\label{fig:board}
\end{figure*}

\subsection{Synthetic buses}
\label{sec:synth}
The repository's generator plants a crossing pattern between two arrays and computes the optimum of its \emph{channel}---every tooth on the source's facing column and every berth on the destination's, in the balls' order---exactly, over every crossing order~\cite{krt}. Twenty-nine such buses of 4 to 16 nets span sorted, blocked, interleaved, shuffled and reversed patterns, short 3\,mm channels, all-via escapes, walls of parts in the channel, rings round the destination, pairs, and lanes that must wind round it. On two layers the whole route completes 23 of them connected and DRC-clean. Of the 16 whose channel optimum applies, it lays exactly the optimum on 10 and completes two more above it, at 18 and 34 vias where 10 and 14 are optimal, with every escape a via. The six it leaves incomplete all have the short channel, and leave one to five nets open: two layers do not give the changes room there. On cases without a graded optimum the route can lay fewer vias than the channel's floor, because it chooses its own ends and may leave an array by another face.

\subsection{More routing layers}
\label{sec:layers}
The model extends to $q>2$ routing layers with few changes. A lane's layer between its changes takes one of $q$ values, two lanes cross only on different layers, and a crossing's room is shared only by the lanes that cross on one layer. An end whose stub carries a via may leave on any layer its run fits, the solve choosing which, and the run is moved onto that layer before the geometry. In the spin picture the layer becomes a $q$-state Potts spin~\cite{wu1982potts} with an antiferromagnetic bond at every crossing, and frustration begins only at $q+1$ mutually crossing lanes: the clique bound of Section~\ref{sec:ends} becomes $c-q$ changes among $c$ such lanes. On the 29 generated buses on three, four and six routing layers, all 87 runs complete connected and DRC-clean, and all 48 whose channel optimum applies lay exactly that optimum, the short channels and all-via escapes included. More layers do not lower these channels' optimum, which equals the floor of Lemma~\ref{lem:floor} on every layer count since the arrays' pads are on F; what they give is room for the changes.

\subsection{Cost}
\label{sec:cost}
Of H3 K51's 290\,s, the fanout with its exact re-ranking takes 42\,s, the solve 32\,s, the geometry LP 150\,s, the polishes 19\,s, the snaps 27\,s, the audits and the lint 11\,s, and laying the plan 6\,s; the geometry is about half the time at the top of both ladders. The time grows a little faster than the bus, from 1.2\,s a net on the 15-net H3 rung to 5.7\,s on the 51-net one, and peak memory is 250 to 720\,MB. On the Zynq board the whole route takes 357\,s of the bus step's 382, 119 of them fanning out both arrays and 122 the geometry. The whole route runs mostly on one core, its CPU time within 11\% of its wall time on every rung, so running three rungs and the board's bus step side by side cost it little waiting. Every rung's first solve is proved optimal: at the top of the H3 ladder its 313 crossings and 2\,744 triples in 30\,s, every other rung's in 7\,s or less, and on the board 229 crossings and 1\,922 triples in 46\,s.

\section{Limitations and Future Work}
\label{sec:limits}
The route does not match lengths; a DDR interface requires length (or delay) matching within byte lanes and of the address group to the clock~\cite{ozdal2006length}. Matching lengths is our next step: as recent length matching assigns each trace a region by linear programming~\cite{lin2026,fang2024}, each lane's meander area can become a variable of the geometry LP. On more than two routing layers the model is so far graded only on generated buses (Section~\ref{sec:layers}), and a pair through a series termination between the arrays is left to the rest of the chain, since the whole route lays no waypoint. Optimality is relative to the model, and the ends are chosen by a local search whose estimate is exact only in its parts; a model over the ends that bounds the route's own objective could prove them. Ends inside the destination's ball field, where Freerouting's lanes enter it, are a natural extension. Letting the ends search swap a DDR controller's data bits within a byte lane, as pin assignment has been chosen jointly with escape routing~\cite{leimak2013}, could remove most crossings within it. With the crossing order fixed and via room set aside, the layer assignment is a planar Ising ground state, solvable in polynomial time~\cite{hadlock1975,barahona1982,glass2026}; it could give the ends search a fast exact evaluator and a lower bound. The solve counts a pair's dive once where the board counts two barrels, and intra-pair skew is not yet planned. Finally, two boards are few, and the design constants were developed on them; we are extending the evaluation to more boards.

\section{Conclusion}
Planning a bus whole---ends, crossings, their order and every layer change decided at once, the plan proved optimal for the chosen ends and negotiated with the geometry until the production router can lay it---routes one real DDR3 bus whole and 43 nets of another fully automatically, and 45 of the second's nets on its own board, with fewer vias and less copper than their designers. The physics proved of practical use: the worldline picture gives the triple rule its meaning, the spin picture traces every via to frustration, and the forces of the linear program's dual name the rules that hold each lane and the ones that must yield. We expect the same pairing---an exact choice of topological sector, relaxed in a convex statics problem whose forces steer the next choice---to carry over to length matching and to boards with more layers.

\bibliographystyle{IEEEtran}
\bibliography{refs}

\end{document}

%% file: figs/lp_cartoon.tex
\begin{tikzpicture}[x=0.78cm,y=0.62cm,font=\scriptsize]
\draw[gray!45,dashed,line width=0.3pt] (0,-0.55)--(0,3.8);
\draw[gray!45,dashed,line width=0.3pt] (1,-0.55)--(1,3.8);
\draw[gray!45,dashed,line width=0.3pt] (2,-0.55)--(2,3.8);
\draw[gray!45,dashed,line width=0.3pt] (3,-0.55)--(3,3.8);
\draw[gray!45,dashed,line width=0.3pt] (4,-0.55)--(4,3.8);
\draw[gray!45,dashed,line width=0.3pt] (5,-0.55)--(5,3.8);
\draw[gray!45,dashed,line width=0.3pt] (6,-0.55)--(6,3.8);
\draw[gray!45,dashed,line width=0.3pt] (7,-0.55)--(7,3.8);
\draw[gray!45,dashed,line width=0.3pt] (8,-0.55)--(8,3.8);
\fill[gray!35] (2.65,-0.5) rectangle (5.35,1.33);
\node[font=\tiny] at (4.00,0.41) {pad box};
\draw[Fcu!85!black,line width=1.1pt] (0,0.000)--(1,0.483)--(2,0.967)--(3,1.450)--(4,1.450)--(5,1.450)--(6,0.967)--(7,0.483)--(8,0.000);
\draw[Fcu!85!black,line width=1.1pt] (0,1.350)--(1,1.717)--(2,2.083)--(3,2.450)--(4,2.450)--(5,2.450)--(6,2.083)--(7,1.717)--(8,1.350);
\draw[Fcu!85!black,line width=1.1pt] (0,2.700)--(1,2.950)--(2,3.200)--(3,3.450)--(4,3.450)--(5,3.450)--(6,3.200)--(7,2.950)--(8,2.700);
\draw[magenta!80!black,line width=1.74pt] (3,1.520)--(3,2.380);
\node[magenta!70!black,right,font=\tiny,inner sep=1pt] at (3,1.950) {0.75};
\draw[magenta!80!black,line width=1.12pt] (3,2.520)--(3,3.380);
\node[magenta!70!black,right,font=\tiny,inner sep=1pt] at (3,2.950) {0.38};
\draw[magenta!80!black,line width=1.74pt] (5,1.520)--(5,2.380);
\node[magenta!70!black,right,font=\tiny,inner sep=1pt] at (5,1.950) {0.75};
\draw[magenta!80!black,line width=1.12pt] (5,2.520)--(5,3.380);
\node[magenta!70!black,right,font=\tiny,inner sep=1pt] at (5,2.950) {0.38};
\draw[orange!90!black,line width=2.70pt,-{Stealth[length=4pt]}] (3,0.600)--(3,1.360);
\node[orange!80!black,right,font=\tiny,inner sep=1.5pt] at (3,0.850) {1.33};
\draw[orange!90!black,line width=2.70pt,-{Stealth[length=4pt]}] (5,0.600)--(5,1.360);
\node[orange!80!black,left,font=\tiny,inner sep=1.5pt] at (5,0.850) {1.33};
\draw[gray,dotted,line width=0.6pt] (1,0.553)--(1,1.647);
\node[gray,right,font=\tiny,inner sep=1pt] at (1,1.100) {0};
\filldraw[fill=black,line width=0.5pt] (0,0.000) circle (2.2pt);
\filldraw[fill=white,line width=0.5pt] (1,0.483) circle (2.2pt);
\filldraw[fill=white,line width=0.5pt] (2,0.967) circle (2.2pt);
\filldraw[fill=white,line width=0.5pt] (3,1.450) circle (2.2pt);
\filldraw[fill=white,line width=0.5pt] (4,1.450) circle (2.2pt);
\filldraw[fill=white,line width=0.5pt] (5,1.450) circle (2.2pt);
\filldraw[fill=white,line width=0.5pt] (6,0.967) circle (2.2pt);
\filldraw[fill=white,line width=0.5pt] (7,0.483) circle (2.2pt);
\filldraw[fill=black,line width=0.5pt] (8,0.000) circle (2.2pt);
\filldraw[fill=black,line width=0.5pt] (0,1.350) circle (2.2pt);
\filldraw[fill=white,line width=0.5pt] (1,1.717) circle (2.2pt);
\filldraw[fill=white,line width=0.5pt] (2,2.083) circle (2.2pt);
\filldraw[fill=white,line width=0.5pt] (3,2.450) circle (2.2pt);
\filldraw[fill=white,line width=0.5pt] (4,2.450) circle (2.2pt);
\filldraw[fill=white,line width=0.5pt] (5,2.450) circle (2.2pt);
\filldraw[fill=white,line width=0.5pt] (6,2.083) circle (2.2pt);
\filldraw[fill=white,line width=0.5pt] (7,1.717) circle (2.2pt);
\filldraw[fill=black,line width=0.5pt] (8,1.350) circle (2.2pt);
\filldraw[fill=black,line width=0.5pt] (0,2.700) circle (2.2pt);
\filldraw[fill=white,line width=0.5pt] (1,2.950) circle (2.2pt);
\filldraw[fill=white,line width=0.5pt] (2,3.200) circle (2.2pt);
\filldraw[fill=white,line width=0.5pt] (3,3.450) circle (2.2pt);
\filldraw[fill=white,line width=0.5pt] (4,3.450) circle (2.2pt);
\filldraw[fill=white,line width=0.5pt] (5,3.450) circle (2.2pt);
\filldraw[fill=white,line width=0.5pt] (6,3.200) circle (2.2pt);
\filldraw[fill=white,line width=0.5pt] (7,2.950) circle (2.2pt);
\filldraw[fill=black,line width=0.5pt] (8,2.700) circle (2.2pt);
\node[above,font=\tiny] at (6,3.300) {$o_{n,k}$};
\node[below,font=\tiny,gray] at (0,-0.55) {pinned};
\node[below,font=\tiny,gray] at (8,-0.55) {pinned};
\end{tikzpicture}